\documentclass[letterpaper, 10 pt, conference]{ieeeconf}  %

\IEEEoverridecommandlockouts                              %

\usepackage{amsmath, amsthm, amssymb,amsfonts,bm}
\usepackage{tikz}
\usepackage{hyperref}
\usepackage{cleveref}
\usetikzlibrary{arrows, decorations.markings}
\usepackage{subcaption}
\usepackage{booktabs}

\usepackage{microtype}
\usepackage{algorithm2e}
\RestyleAlgo{ruled}

\usepackage{siunitx}
\DeclareMathOperator*{\argmax}{arg\,max}
\DeclareMathOperator*{\argmin}{arg\,min}
\tikzstyle{vecArrow} = [thick, decoration={markings,mark=at position
   1 with {\arrow[semithick]{open triangle 60}}, mark=at position 0 with {\arrow[semithick]{Bar}}},
   double distance=2.5pt, shorten >= 5.5pt,
   preaction = {decorate},
   postaction = {draw,line width=1.7pt, white,shorten >= 4.5pt}]
\tikzstyle{innerWhite} = [semithick, white,line width=1.4pt, shorten >= 4.5pt]

\usepackage[table]{xcolor}

\newcommand{\R}{\mathbb{R}}
\newcommand{\bmat}[1]{\begin{bmatrix}#1\end{bmatrix}}
\newtheorem{proposition}{Proposition}

\title{\LARGE \bf
Simulating Arbitrage Optimization for Market Monitoring in Gas and Electricity Transmission Networks
}

\author{Noah Rhodes$^{1}$, Sachin Shivakumar$^{1}$, Luke S. Baker$^{2}$, Kaarthik Sundar$^{3}$, and Anatoly Zlotnik$^{4}$ %
\thanks{*This project was supported by the LDRD program and the Center for Nonlinear Studies at Los Alamos National Laboratory, as well as the U.S. DOE Office of Electricity Project TE1101000-05300-3123789: Fuel Reliability for Electric Energy Delivery by Optimized Management of Gas-pipeline Automation Systems (FREEDOM-GAS).  Research conducted at Los Alamos National Laboratory is done under the auspices of the National Nuclear Security Administration of the U.S. Department of Energy under Contract No. 89233218CNA000001. Report No. LA-UR-26-22589.}%
\thanks{$^{1}$Noah Rhodes and Sachin Shivakumar are with the Center for Nonlinear Studies at the Los Alamos National Laboratory, Los Alamos, NM 87545 USA {\tt\small \{nrhodes,sshivakumar\}@lanl.gov}}%
\thanks{$^{2}$Luke S. Baker is in the Accelerator Operations \& Technology Division at the Los Alamos National Laboratory, Los Alamos, NM 87545 USA {\tt\small lsbaker@lanl.gov}}%
\thanks{$^{3}$Kaarthik Sundar is in the Information Systems \& Modeling Group at the Los Alamos National Laboratory, Los Alamos, NM 87545 USA {\tt\small kaarthik@lanl.gov}}%
\thanks{$^{4}$Anatoly Zlotnik is in the Applied Mathematics \& Plasma Physics Group at the Los Alamos National Laboratory, Los Alamos, NM 87545 USA {\tt\small azlotnik@lanl.gov}}%
}

\allowdisplaybreaks

\begin{document}

\maketitle
\thispagestyle{empty}
\pagestyle{empty}

\begin{abstract}
We examine market outcomes in energy transport networks with a focus on gas-fired generators, which are producers in a wholesale electricity market and consumers in the natural gas market.  Market administrators monitor bids to determine whether a participant wields market power to manipulate the price of energy, reserves, or financial transmission rights. If economic or physical withholding of generation from the market is detected, mitigation is imposed by replacing excessive bids with reference level bids to prevent artificial supply shortages. We review market monitoring processes in the power grid, and present scenarios in small interpretable test networks to show how gas-fired generators can bid in the gas market to alter outcomes in a power market.  We develop a framework based on DC optimal power flow (OPF) and steady-state optimal gas flow (OGF) formulations to represent two interacting markets with structured exchange of price and quantity bids.  We formulate optimization-based methods to identify market power in a power grid, as well as to identify market conditions that indicate market power being exerted by a generator using gas market bids.    
\end{abstract}

\section{INTRODUCTION}

The rapid construction of gas-fired generators to compensate for the variability of renewable electricity and the retirement of coal-fired power plants has created interdependence between the physical flows in the power grid and gas pipelines as well as markets for wholesale electricity and natural gas
\cite{he2018coordination}.   This trend motivates methods for physical and financial coordination of power and gas delivery in day-ahead markets and in real-time \cite{guerra2021coordinated}, and the economic value of such coordination to energy users was shown \cite{craig2020valuing}.  Most gas-electric system coordination studies have treated gas-fired generators as passive elements that couple  power grids and gas pipelines through a heat-rate curve, which may be quadratic or even linear \cite{zlotnik2016coordinated,guerra2021coordinated}, with marginal costs that are fixed or adjusted according to a linear demand curve \cite{duenas2014gas}.

As the deregulation of energy markets worldwide has transformed natural monopolies into competitive markets \cite{bunn1997modelling}, the potential for strategic imbalances has grown.  Gas-fired generators are increasingly owned and operated as large fleets that participate as producers in  electricity markets and as consumers in gas and gas transportation markets.  Generation fleet managers thus have incentives and exercise rationality in maximizing their profits, which are tied to the spread between the price of electricity in the real-time market (RTM) and natural gas in the secondary release capacity market (SRCM) \cite{nandakumar2016computational}.  This motivates  advanced modeling of combined-cycle power plant operations and market participation \cite{leisen2024modeling}, as well as the nuanced effect of their roles on market equilibria and market design, particularly in the context of increased renewables in the generation mix \cite{levin2015electricity,guo2018market}.  Bi-level strategic bidding models have been proposed to examine the market behaviors of individual gas-fired units, where the upper level problem maximizes its profit and two lower-level problems model the electricity and natural gas markets \cite{jiang2021bi}.

Gas pipelines play a critical role in the energy systems in the U.S. and Europe 
\cite{mohlin2021us,hulshof2016market}.  The interdependence with power grids results in complex market dynamics \cite{jenkins2015dynamic}, in which participants in the natural gas market could potentially exert market power to influence outcomes in the electricity market \cite{spiecker2011analyzing}.  This dynamic was modeled using a mixed complementarity problem in the form of a Cournot-Nash equilibrium subject to some regularity assumptions, and where the outcome is highly influenced by the bidding behavior in the gas market \cite{spiecker2013modeling}.  The latter study focused on potential behaviors of gas producers and traders while considering gas-fired generators as passive components, and similar assumptions and models were used to show evidence for strong incentives for vertical integration of gas and electricity markets \cite{morey2025winners}.  In many regions, the markets for gas and gas transportation are decoupled in the regulatory framework, with the latter highly regulated.  Transportation markets typically include balancing rules, intended to maintain a balance in physical flows while providing hourly, daily, or even monthly flexibility and market liquidity \cite{creti2016natural,TGP_EntireTariff}.  Regulation of pipeline flexibility can create significant market distortions by reducing effective capacity \cite{keyaerts2011gas}, and can lead to complexity in operations.  Little guidance on gas market monitoring is available either from EU institutions or in academic literature  \cite{szolnoki2017monitoring}, and the need for improved monitoring of pipelines and gas markets in the U.S. has been emphasized.

Because independent system operators (ISOs) in the power sector administer integrated markets for power production and transmission, robust market monitoring frameworks and market power mitigation mechanisms are in place \cite{graf2021market}, including at CAISO \cite{CAISO_mm}, PJM \cite{pjm2025som}, NYISO \cite{NYISO2023SOMreport}, MISO \cite{MISO_MPB_Energy_Market}, and ERCOT \cite{ercot2025mitigation}.  Market monitoring methods were introduced to address often controversial and complex challenges following electricity market deregulation in many parts of the world in the 1990's \cite{goldman2004review,garcia2007international}. Regulators use statistical methods to assess the ability of a seller to profitably maintain prices above competitive levels for significant periods of time measured in years using economic or physical withholding \cite{rahimi2003effective}, in addition to other potential instances of market power exertion \cite{adib2008market,peterson2001best}.  In addition to anti-competitive practices, optimization-based energy markets can result in negative wholesale electricity prices \cite{sun2026negative}, and the empirical irregularities caused by short term pricing continue to be studied by researchers and addressed in practice \cite{barbour2014can,stiewe2025cross}.

In this study, we formalize market monitoring for anti-competitive wholesale electricity bidding in the DC OPF problem, and show how market power can be identified by examining the OPF problem in dual space.  We use the optimal gas flow (OGF) problem to represent an efficient gas market, and use examples for small networks to illustrate market power scenarios in coupled gas and electricity systems.  We formulate an optimal arbitrage problem for a gas-fired generator that participates in wholesale gas and electricity markets and seeks to maximize its profits. In contrast to previous studies on gas-electric network markets, we suppose that the collections of bids into these markets are not accessed by a single administrator so that the OGF and OPF problems cannot be combined in a single optimization formulation. Towards developing a monitoring framework to detect anomolous bidding, we iteratively solve OPF and OGF problems to identify a profit-maximizing bid structure which could be used to determine conduct thresholds.

The rest of the paper is as follows.  Section \ref{sec:market} reviews market monitoring and mitigation techniques based on recent academic studies and published industry standards.  In Section \ref{sec:arbitrage}, we describe the optimal profit problem for gas-fired generators that arbitrage across the so-called financial heat rate that connects electricity and gas transmission networks.  In Section \ref{sec:examples}, we examine market power exertion in electricity and coupled gas-electricity networks using scenarios for small, interpretable test systems.  We conclude with a discussion of policy implications in Section \ref{sec:conc}.

\section{MARKET MITIGATION}  \label{sec:market}

The primary functions of a market monitoring division at an electricity ISO are to examine \cite{CAISO_mm}:
\begin{itemize}
    \item potential anti-competitive behavior or inefficiencies;
    \item ineffective rules or practices and changes to improve wholesale competition and efficient outcomes;
    \item bid mitigation rules to remedy the potential exercise of market power; and
    \item effectiveness of the market in signaling for investment in generation, transmission, and demand response.
\end{itemize}
Effective market monitoring requires accessible and comprehensive data as well as institutional independence \cite{twomey2005review}.  We review the market mitigation rules enforced in Midcontinent Independent System Operator (MISO) below \cite{MISO_MPB_Energy_Market}.

\vspace{-1ex}
\subsection{Review of Market Mitigation Rules} \label{subsec:misorules}
Market mitigation is applied by an administrator when a participant is determined to be wielding market power to manipulate the price of energy, reserves, or financial transmission rights (FTRs).  Market power is seen to be wielded when the participant submits bids with artificially low or high prices that do not represent the marginal price, or bids with artificial physical values like ramp rate or minimum generation level.  These values represent economic or physical withholding of generation that create artificial supply shortages. Because market designs influence the strategies that participants use, market mitigation must be tailored to the specific short-term market design.

The purpose of mitigation is to prevent the exertion of market power, particularly in constrained regions of the network that are more exposed to manipulation due to fewer participants.
Mitigation measures are only applied in response to one or more binding transmission or reserve zone constraints.  If there are no binding network constraints, then it is assumed that many market participants are competing and no mitigation is required.  

Two steps are used to test if market power is being exerted. First, a Conduct Test is used to identify offers that exceed reference levels above a threshold. Next, an Impact Test determines what effect the offer had on the market clearing, and mitigation is imposed if the impact exceeds a threshold.  Market mitigation is the replacement of the bid with a reference level bid \cite{goldman2004review}.  The Conduct and Impact Test thresholds depend on the area of the network it occurs in.  Definitions of constrained areas (CAs) vary among ISOs, and can be pre-specified or dynamically defined by market conditions according to factors such as i) the long-term likelihood of binding transmission or reserve constraints over the past year; ii) the presence of a pivotal supplier in the region; iii) the likelihood of short-term binding transmission constraints over the past week; and iv) excessive Generator Shift Factors (GSFs), which are flow contributions that a change in generator production causes in identified transmission lines.  Various conduct tests are applied, including thresholds for physical withholding of resources, planning resources, and transmission lines; for economic withholding for planning resources, in reserve zones, and in constrained areas; and for uneconomic production.  We consider economic withholding in narrow and broad constrained areas.

\emph{Narrow Constrained Area (NCA)}: An NCA is an electrical area that \emph{has at least one Binding Transmission Constraint or Binding Reserve Constraint into or in a common electrical area in the ISO that is expected to be binding for more than 500 hours in a given 12 month period \textbf{and} if at least one supplier is pivotal in that electrical region}.  Because network constraints frequently bind in this region, market power is likely to be substantial.

\emph{Broad Constrained Areas (BCA)}: A BCA is an electrical area where sufficient competition exists such that a binding network constraint generally does not result in substantial market power. A BCA is not predefined area, and depends on system Generator Shift Factors (GSFs), which are the flow contributions that a change in generator production imposes on given transmission lines.  In a BCA, the Conduct Test compares the Constraint Generation Shift Factor Cutoff (CGSFC) of the constrained line to the value of a GSF, and if the differences exceeds a limit (ranging from \textpm 3\% to 6\% depending on voltage), then a resource is determined to have a significant effect on congestion.  A resource faces a conduct test in a BCA if \emph{there is a positive resource GSF that is greater than the BCA positive CGSFC \textbf{OR} A negative resource GSF that is less than the BCA negative CGSFC.}

\vspace{-1ex}
\subsection{Economic Withholding Conduct Test}  \label{subsec:conducttest}

A conduct test determines whether generator actions are the result of scarcity in the power system, or if they might be the result of market power exertion.

\emph{Economic withholding conduct thresholds in NCAs}:
All suppliers in an NCA are subject to a conduct test when a relevant network constraint is binding. The NCA conduct threshold is $\text{Con}_{NCA}= [\text{Net Annual Fixed Cost}]/[\text{Constrained Hours}]$,  where the Net Annual Fixed Cost is for new generation in \$/MW, and Constrained Hours is the total hours in the last 12 months that an interface to the NCA is binding, not exceeding 2000 hours.  As of June 6, 2025, these value range from \$22.05 to \$100.00 in MISO. %
A supplier fails the Conduct Test if an \emph{Energy offer or Minimum Generator offer is above the Reference Level by more than the NCA Threshold}.

\emph{Economic withholding conduct thresholds in BCAs}:
In a BCA, a generator fails a conduct test if an \emph{Energy Offer or Minimum Generation offer is a 300\% increase  or is 100\$/MWh above the Reference Level (offers below 25\$/MWh are not subject)}.

The comparison Reference Level is a benchmark for the Conduct Test. There are two methods to determine a Reference Level. 1) Offer-Based: The lower of the mean or the median of a Resource's accepted Energy Offers in competitive periods over the previous 90 days separately for On-Peak and Off-Peak  periods, and for Real-Time and Day-Ahead markets, adjusted for changes in spot fuel prices.  2) Cost-Based:  A consultation intended to reflect the Resource's marginal cost including risk premiums and opportunity costs.  A Cost-Based method is used when the resource is not dispatched frequently enough in the past 90 days to calculate an Offer-Based Reference Level. 

\vspace{-1ex}
\subsection{Economic Withholding Impact Test}  \label{subsec:impacttest}

If a resource fails a Conduct Test, an Impact Test checks if the offer that failed the conduct test affected the market outcome.   The market outcome with the generator offers is compared against a market outcome with Reference Level offers substituted for the offers that fail the Conduct Test. 

A generator fails the Impact Test in an NCA if \emph{ the LMP is increased by more than $\text{Imp}_{NCA}=[\text{Net Annual Fixed Cost}]/[\text{Constrained Hours}]$}. A generator fails the Impact Test in a BCA if \emph{the LMP is increased by 200\% or \$100/MWh, whichever is lower}.  If a generator fails the Impact Test, its offer is replaced with a Reference Level offer.   If a resource fails, all other resources from the owner in the same NCA or BCA are mitigated.  If the owner fails an impact test in the next 90 days in the same NCA or BCA, their resources are mitigated for the next 7 days. 

ISO-NE and NYISO use a similar ``conduct and impact'' approach, while CAISO and PJM use a ``structural'' approach with pre-determined system conditions and checks for congestion relief resources to test if a constraint is competitive \cite{graf2021market}.  ERCOT uses a unique hybrid approach \cite{graf2021market}.

\subsection{Formal Definition of Market Mitigation}   \label{subsec:formal}

For a network consisting of nodes $\mathcal{N}$ and lines $\mathcal{L}$, with generators $\mathcal{G}$ and buses $\mathcal{B}$, there is a solution to an optimal power flow problem $x^*$ that minimizes the cost of delivering power to the loads given the offer bids of generators.  Assume that some of the flows across lines $\mathcal{L}_b\subset\mathcal{L}$ are binding .  A subnetwork $\mathcal{N}_b\subset\mathcal{N}$ is a constrained area if an increase in load in this region can only be served by generators in this region.  The generators in this region are $\mathcal{G}_b$.

Each generator has a reference marginal cost $c^R_g$ and a bid cost $c^B_g$ that it submits to the market.  A conduct test is done for each generator in $\mathcal{G}_b$, which fails if its bid exceeds the reference by either \$100 or a factor of 4:
\begin{equation} \label{eq:conducttest}
CT \gets c^B_g \ge 4c^R_g \,\,\, || \,\,\, c^B_g \ge \$100 + c^R_g
\end{equation}
If a generator fails the conduct test, it undergoes an impact test. For each generator, its bid is replaced with the reference level and a market clearing problem is solved.  A generator fails the impact test if the original LMP $\lambda^B_g$ is 3 times the reference LMP $\lambda^R_g$, or \$100 higher than the reference LMP:
\begin{equation} \label{eq:impacttest}
 \!\! IT \gets \lambda^B_g \ge 3\lambda^R_g \,\,\, || \,\,\, \lambda^B_g \ge \$100 + \lambda^R_g
\end{equation}
If a generator fails the impact test, its bid price is replaced by the reference price.  This action is called mitigation, and is taken because these tests demonstrate the that the generator was successfully exerting market power to increase the price of power in the constrained region.

\section{OPTIMAL GENERATION PROFIT ARBITRAGE} \label{sec:arbitrage}
We next explore how a generator might optimize its bidding in gas and electricity markets to increase its profit.  We begin by introducing the DC Optimal Power Flow (OPF) and the Optimal Gas Flow (OGF) problems, which we use to represent efficient electrical and gas markets without accounting for time-dependent factors. Then we introduce the generator bidding problem, where the generator identifies a bidding strategy to maximize its profit in a way that does not trigger the electric market monitor rules.

\subsection{OPF: Optimal Power Flow}  \label{sec:opf}

The OPF problem minimizes electricity production cost while meeting demand \cite{caramanis2007optimal}.  The objective in \eqref{eq:opf_obj} is production cost, which is the sum of a generator bid price $c_g$ times the power it produces $\boldsymbol{P}_g$. In eq. \eqref{eq:opf_power_blance}, the power generated, consumed, and flow in and out of node  must sum to 0.
\begin{subequations} \label{opt:OPF}
\begin{align}
 \min\limits_{\boldsymbol{P}, \boldsymbol{\theta}} &\quad \sum_{g\in \mathcal{G}} c_{g}^P \boldsymbol{P}^G_{g} \label{eq:opf_obj}\\
 \text{s.t.,} \qquad&  \Omega\boldsymbol{P}^G + \Psi\boldsymbol{\theta} - P^D = 0 && \label{eq:opf_power_blance}\\ 
 & -\overline{T} \le \Phi\boldsymbol{\theta} \le \overline{T} &&  \label{eq:opf_thermal_limit}\\
 & \underline{P}^G \le \boldsymbol{P}^G \le \overline{P^G} &&  \label{eq:opf_gen_limit}
\end{align}
\end{subequations}
 $\Omega\in \{0,1\} ^{|\mathcal{B}|\times |\mathcal{G}|}$ is the generator-bus incidence matrix, $\Psi\in\R^{|\mathcal{B}|\times |\mathcal{B}|}$ is the network admittance matrix, $\Phi\in \{0,1
 \}^{|\mathcal{L}|\times |\mathcal{B}|}$ is the line-bus incidence matrix, $\boldsymbol{\theta}\in\R^{|\mathcal{B}|}$ is a vector of nodal voltage angles, and $P^D\in \R^{|\mathcal{B}|}$ is a vector of power demands, $\mathcal{B}$ is set of all buses, and $\mathcal{G}$ is set of all generators. Eq. \eqref{eq:opf_thermal_limit} and \eqref{eq:opf_gen_limit} constrain maximum line power to $\overline{T}\in \R^{|\mathcal{L}|}$ where $\mathcal{L}$ is set of all lines, and generator output to between a minimum of $\underline{P_g}$ and a maximum of $\overline{P_g}$.

Given the above OPF formulation, a simple linear program (LP) feasibility test can verify market power. Consider a subset of generators $\mathcal{G}^*\subset \mathcal{G}$ and solve a generator cost maximization problem subject to the OPF constraints:
\begin{align}
    \label{eq:OPF_gen}
    \max\limits_{(c^P_g, \underline{P}_g,\bar{P}_g)_{g\in\mathcal{G}^*}} \quad & \sum\limits_{g\in \mathcal{G}^*} c^P_g \boldsymbol{P}_g,\notag\\
    \text{s.t.,} \quad & (\boldsymbol{P},\boldsymbol{\theta})~\text{satisfy Eq. \eqref{eq:opf_power_blance}-\eqref{eq:opf_gen_limit}}.
\end{align}
Heuristically, maximizing the objective in Eq. \eqref{eq:OPF_gen} is equivalent to maximizing the objective in Eq. \eqref{eq:opf_obj}, assuming the parameters corresponding to generators in $\mathcal{G}\setminus \mathcal{G}^*$ are fixed. This leads to a $\max-\min$ optimization of the form
\begin{align}\label{eq:OPF_bilevel}
    \max_{x_c,x_l} \min_y \quad & \bmat{\bmat{\bar{x}^\top& x_c^\top}& 0} y,\notag\\
    \text{s.t.,} \quad & A_{eq} y =b_{eq}, \; A_{ineq} y \le B_{ineq} x_l,
\end{align}
where $\bar{x} \!=\! c^P_{g\in\mathcal{G}\setminus \mathcal{G}^*}$, $x_c \!=\! c^P_{g\in\mathcal{G}^*}$, $x_l\!=\!(I,\underline{P}_{g\in\mathcal{G}^*},\bar{P}_{g\in\mathcal{G}^*})$  and $y\!=\!(\boldsymbol{P}_{g\in\mathcal{G}\setminus\mathcal{G}^*},\boldsymbol{P}_{g\in\mathcal{G}^*}, \boldsymbol{\theta}_{i\in \mathcal{B}})$ are sets of generator decision variables and inner OPF variables, respectively. 
For appropriate network parameters,
let
\begin{align*}
    A_{eq} &= \bmat{\Omega& \Psi},&& b_{eq} = P^D,\\
    A_{ineq} &= \bmat{0&\Phi\\0&-\Phi\\I&0\\-I&0}, && B_{ineq} = \bmat{\overline{T}&0&0\\\overline{T}&0&0\\0&-I&0\\0&0&I}.
\end{align*}
Using optimality conditions of the inner optimization problem in Eq. \eqref{eq:OPF_bilevel} yields a single level formulation for Eq. \eqref{eq:OPF_bilevel}:
\begin{align}\label{eq:OPF_singlelevel_primal}
\max_{x_c,x_l,\lambda,\mu} \quad &b_{eq}^\top\lambda +x_l^\top B_{ineq}^\top \mu,\notag\\
\text{s.t.,} \quad & \bmat{\bar{x}\\x_c} + A_{eq}^\top \lambda +A_{ineq}^\top\mu = 0, \; \mu\ge 0.
\end{align}
Given bounds $x_l$ on generation limits, Eq. \eqref{eq:OPF_singlelevel_primal} is an LP and its dual is given by
\begin{align}\label{eq:OPF_singlelevel_dual}
    \min_y \quad & \bmat{\bmat{\bar{x}^\top&0}&0} y,\notag\\
    \text{s.t.,} \quad & A_{eq} y =b_{eq}, \; A_{ineq} y \le B_{ineq} x_l,\notag\\
    & \bmat{\bmat{0& I}&0} y =0.
\end{align}
In problem \eqref{eq:OPF_gen}, we suppose that generators can control their offer price and production limits.

\begin{proposition}\label{prop:marketPowerTest}
    For a given $x_l$, suppose Eq. \eqref{eq:OPF_singlelevel_primal} is feasible and Eq. \eqref{eq:OPF_singlelevel_dual} is infeasible. Then, the objective of Eq. \eqref{eq:OPF_bilevel} is unbounded above with respect to the decision variable $x_c$.
\end{proposition}
\begin{proof}
    The proof follows from the duality in LPs. Because Eqs. \eqref{eq:OPF_singlelevel_primal} and \eqref{eq:OPF_singlelevel_dual} form a primal-dual pair, infeasibility of Eq. \eqref{eq:OPF_singlelevel_dual} implies unboundedness of the primal objective if the primal is feasible. Because Eqs. \eqref{eq:OPF_bilevel} and \eqref{eq:OPF_singlelevel_primal} are equivalent for fixed $x_l$, it follows that the objective of Eq. \eqref{eq:OPF_bilevel} is  unbounded.
\end{proof}
\noindent Proposition \ref{prop:marketPowerTest} implies that, given a subset of generators $\mathcal{G}^*$ and generation limits for all members of $\mathcal{G}$, one can check whether members of $\mathcal{G}^*$ have market power (i.e., to set arbitrary cost of generation) or not with an LP feasibility test.

\vspace{-1ex}
\subsection{OGF: Optimal Gas Flow}  \label{sec:ogf}

Pipelines transport natural gas from suppliers to consumers.  In steady-state, the Weymouth equation relates the pressures at the ends of a pipe to the mass flow through the pipe \cite{rios2015optimization}.   The steady-state optimal gas flow (OGF) problem was developed to maximize economic utility for gas pipeline users subject to nonlinear physics and engineering constraints \cite{rudkevich2017locational}.  The controllable parameters of the OGF are any loads that can be curtailed, as well as the settings of gas compressors. 

Consider a pipeline as a directed graph with sets of physical nodes $\mathcal N_G$ and edges $\mathcal E_G$ that consist of pipes $\mathcal P_G$ and compressors $\mathcal C_G$. Each pipe and compressor connects two distinct nodes and has defined positive direction of mass flow. Edges are denoted as $(i,j) \in \mathcal E_G$ for a pipe or compressor directed from node $i \in \mathcal N_G$ to node $j \in \mathcal N_G$. The key outcomes for market participants $\mathcal M_G$ are injections and withdrawals that occur at some physical nodes $\mathcal N_G$. We suppose that the set of market nodes that are associated to a physical node $j \in \mathcal N_G$ is called $\partial_j\subset\mathcal M$.  For each market participant $m \in \mathcal M$, the maximum  supply and delivery nominations are $s_m^{\max}$ and $d_m^{\max}$. The actual quantities supplied $s_m$ and delivered $d_m$ are decision variables. The offer and bid prices $c_m^s$ and $c_m^d$ are parameters.  Decisions are defined by a nonlinear optimization problem.

We describe network flow by the Weymouth equation in Eq. \eqref{eq:gas_flow}, mass balances at nodes in Eq. \eqref{eq:flow_balance}, and the action of compressors in Eq. \eqref{eq:comp_action}.  
\begin{subequations} \label{eq:gas_equality}
    \begin{eqnarray}
 p_j^2 - p_i^2  = -\beta_{ij} |\varphi_{ij}|\varphi_{ij}, \quad \forall (i,j)\in\mathcal P_G \label{eq:gas_flow}\\
  \sum_{(ij)\in \mathcal E_G} \varphi_{ij}-\sum_{(jk)\in \mathcal E_G} \varphi_{jk} =\sum_{m\in \partial_j} \left( s_{m}-  d_{m} \right), \,\, \forall j\in\mathcal N, \label{eq:flow_balance} \\
  p_j = \alpha_{ij} p_i, \quad \forall (i,j)\in\mathcal C_G. \label{eq:comp_action}
\end{eqnarray}
\end{subequations}
Here $p_i$ is the gas pressure at a node $i\in\mathcal N_G$, $\varphi_{ij}$ is the mass flow on pipe $(i,j)\in\mathcal E$, $s_m$ and $d_m$ are injections or withdrawals from the network at market nodes $m\in\partial_j$, and $\alpha_{ij}$ is the pressure boost of a gas compressor $(i,j)\in\mathcal C_E$.   In the Weymouth equation \eqref{eq:gas_flow}, the parameter $\beta_{ij}$ is called the \emph{resistance} and is provided for each pipe by $\beta=a^2\lambda L/ (A^2 D)$, where $a$ denotes the local wave speed in the gas, $\lambda$ is the Darcy-Weisbach friction factor, and $L$, $D$, and $A$ denote the length, diameter, and cross-sectional area of the pipe, respectively. We suppose $a$ is constant and uniform, while each pipe $(i,j)\in\mathcal{E}_g$ has parameters $\lambda_{ij}$, $L_{ij}$, $D_{ij}$, and $A_{ij}$.  Engineering limits and nomination bounds have the form
\begin{subequations} \label{eq:gas_inequality}
    \begin{align}
1 &\le \alpha_{ij} \le \alpha_{ij}^{\max}, &(i,j)&\in \mathcal C_G, \label{eq:comprs_inequality} \\
        p_j^{\min}  & \le p_j \le p_j^{\max}, &j&\in \mathcal N_G, \label{eq:pres_inequality} \\
      0  & \le  s_{m} \le  s_{m}^{\max}, &m&\in \mathcal M_G, \label{eq:supply_inequality} \\
        0 & \le  d_{m} \le  d_{m}^{\max}, &m&\in \mathcal M_G. \label{eq:demand_inequality}
\end{align}
\end{subequations}
Equation \eqref{eq:comprs_inequality} constrains the compressor to a maximum pressure boost, Eq. \eqref{eq:pres_inequality} constrains the pressure to operating bounds, and Eqs. \eqref{eq:supply_inequality}-\eqref{eq:demand_inequality} constrain the injections and withdrawals to within $s_m^{\max}$ and $d_m^{\max}$.
The objective  is 
\begin{equation} \label{eq:gas_objective}
 J_G\!=\!\!\sum_{m\in \mathcal M}\!\!\! \left(c^d_{m}  d_{m} \!-\! c^s_{m}  s_{m} \right) - \!\!\!\sum_{(i,j)\in \mathcal C_G} c^C_{ij}\varphi_{ij} \left(\alpha_{ij}^r\!\!-\!1 \right)\!,
\end{equation}
where $r$ is related to the adiabatic compression equation \cite{rios2015optimization}.
The first summation represents total market surplus, and the second summation represents the energy used to power compressors. The OGF problem is given as
\begin{subequations}
\begin{align}
 \min_{p,\varphi,\alpha,s,d} \quad & J_G \text{ in } \eqref{eq:gas_objective}  \label{eq:gas_obj}\\
 \text{s.t.,} \quad & \text{gas network flow } \eqref{eq:gas_equality} \\
  & \text{gas network inequality } \eqref{eq:gas_inequality}.
\end{align}
\end{subequations}
Unlike the OPF in Eq. \eqref{opt:OPF}, the OGF is non-convex for which simple market-power tests like Prop. \ref{prop:marketPowerTest} are not available. Hence, we seek a heuristic to couple the OPF and OGF problems in a unified generator price arbitrage (GPA) problem. Next, we formulate a bilevel optimization formulation of GPA and then discuss a possible iterative approach for solving it.

\subsection{GPA: Generator Price Arbitrage}  \label{sec:gpb}

In the \emph{Generator Price Arbitrage} problem, a fleet of gas-fired generators, indexed by $g\in\mathcal{G}^*$, place price-quantity bids into power and gas markets. The OPF and OGF are treated as black-boxes from the generator's perspective, sending back only the cleared price-quantity pairs (Fig. \ref{fig:coupled_market}). 

Because the generators search for optimal bids to maximize revenue, the decision variables are price-quantity bids into the two markets, while the objective is the revenue of power sold offset by cost of gas procured.  The GPA is then
\begin{align}\label{opt:GPA}
\max\limits_{ (\bar{c}^P_{i}, \bar{\boldsymbol{P}}_i, \bar{c}^G_{i}, \bar{\boldsymbol{q}}_i)_{i\in\mathcal{G}^*}} &\quad\sum\limits_{i\in \mathcal{G}^*} \left(c^P_{i}\boldsymbol{P}_{i} - c^G_{i} \boldsymbol{q}_i \right),\\
\text{s.t.,}\qquad& (c^P_{\mathcal{G}^*}, \boldsymbol{P}_{\mathcal{G}^*}) = \argmin \,\mathrm{OPF}(\bar c^P_{\mathcal{G}^*}, \bar{\boldsymbol{P}}_{\mathcal{G}^*}; \delta_1),\notag\\
& (c^G_{\mathcal{G}^*}, \boldsymbol{q}_{\mathcal{G}^*}) = \argmax \, \mathrm{OGF}(\bar c^G_{\mathcal{G}^*}, \bar{\boldsymbol{q}}_{\mathcal{G}^*}; \delta_2),\notag\\
& \bar{\boldsymbol{P}}_{\mathcal{G}^*} = h(\boldsymbol{q}_{\mathcal{G}^*}),\, \bar{c}^P_{\mathcal{G}^*}\in\mathcal{F}^P_{CT}, \, c^P_{\mathcal{G}^*}\in\mathcal{F}^P_{IT}, \notag
\end{align}
where price-quantity bids into the electricity market $(\bar{c}^P_{\mathcal{G}^*}, \bar{\boldsymbol{P}}_{\mathcal{G}^*})$ and gas market $(\bar{c}^G_{\mathcal{G}^*}, \bar{\boldsymbol{q}}_{\mathcal{G}^*})$ are decision variables, $\delta_i$ denote market parameters unknown to members of $\mathcal{G}^*$, production capacity bids $\bar{\boldsymbol{P}}_{\mathcal{G}^*}$ are limited by pass-through heat-rate curve $h(\boldsymbol{q}_{\mathcal{G}^*})$, and $\mathcal{F}^P_{CT}$ and $\mathcal{F}^P_{IT}$ represent bids that do not trigger the impact or conduct tests in equations \eqref{eq:conducttest}-\eqref{eq:impacttest}. 
Solving the OPF and OGF yields cleared price-quantity pairs $(c^P_{\mathcal{G}^*}, \boldsymbol{P}_{\mathcal{G}^*})$ and $(c^G_{\mathcal{G}^*}, \boldsymbol{q}_{\mathcal{G}^*})$.

The GPA problem \eqref{opt:GPA} is ill-posed due to incomplete information and optimization complexity, and hence a closed-form solution is unavailable. Here we develop an iterative heuristic to obtain feasible local candidate solutions.

\begin{figure}[!ht]
\centering
\vspace{-2ex}
\begin{tikzpicture}

\draw[vecArrow] (0.5,1.75) -- node[midway, above] {$c^P_{\mathcal{G}^*}, \boldsymbol{P}_{\mathcal{G}^*}$} (2.25,1.75);
\draw[vecArrow] (2.25,0.25) -- node[midway, below] {$\bar c^P_{\mathcal{G}^*}, \bar{\boldsymbol{P}}_{\mathcal{G}^*}$} (0.5,0.25);
\draw[vecArrow] (3.75,1.75) -- node[midway, above] {$\bar c^G_{\mathcal{G}^*}, \bar{\boldsymbol{q}}_{\mathcal{G}^*}$} (5.5,1.75);
\draw[vecArrow] (5.5,0.25) -- node[midway, below] {$c^G_{\mathcal{G}^*}, \boldsymbol{q}_{\mathcal{G}^*}$} (3.75,0.25);

\draw (-0.5,0) rectangle (0.5,2);

\draw (2.25,0) rectangle (3.75,2);

\draw (5.5,0) rectangle (6.5,2);

\node (OPF) at (0.,1) {OPF};
\node (Gen) at (3,1.25) {Generator};
\node at (3,0.75) {$\{i\in \mathcal{G}^*\}$};
\node (OGF) at (6,1) {OGF};

\end{tikzpicture}
\caption{\small Diagram gas-fired generator bids in the OPF and OGF.} 
\label{fig:coupled_market}
\end{figure}
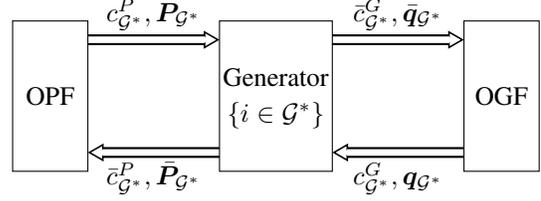

\vspace{-3ex}
\subsection{Iterative Bidding Strategy} \label{subsec:iterativebid}

For the iterative scheme, we assume that the electricity and gas markets are cleared at different times so generators can use additional information from the dispatch of one market to update the bids into the other.
Algorithmically, this translates to the following steps: For the $0$th iteration, initial bids for the electricity market are based on some reference $(\bar c_{\mathcal{G}^*}^{P,0}, \bar{\boldsymbol{P}}_{\mathcal{G}^*}^0)$.
\begin{enumerate}
    \item Given $(\bar c_{\mathcal{G}^*}^{P,k}, \bar{\boldsymbol{P}}_{\mathcal{G}^*}^k)$ solve the OPF problem to obtain cleared price-quantities: $(c_{\mathcal{G}^*}^{P,k}, {\boldsymbol{P}}_{\mathcal{G}^*}^k)$
    \item Using $(c_{\mathcal{G}^*}^{P,k}, {\boldsymbol{P}}_{\mathcal{G}^*}^k)$, set $(\bar c_{\mathcal{G}^*}^{G,k}, \bar{\boldsymbol{q}}_{\mathcal{G}^*}^k)$ such that $\bar{\boldsymbol{q}}_{\mathcal{G}^*}^k = h^{-1}({\boldsymbol{P}}_{\mathcal{G}^*}^k)$ and $\sum_{i\in \mathcal{G}^*} \left(c_{i}^{P,k}\boldsymbol{P}_i^k - \bar c_{i}^{G,k} \bar{\boldsymbol{q}}_i^k \right)\ge 0$.
    \item Solve the OGF using bids $(\bar c_{\mathcal{G}^*}^{G,k}, \bar{\boldsymbol{q}}_{\mathcal{G}^*}^k)$ to obtain cleared price-quantities $(c_{\mathcal{G}^*}^{G,k}, {\boldsymbol{q}}_{\mathcal{G}^*}^k)$
    \item If $\boldsymbol{q}^k< \bar{\boldsymbol{q}}^k$, reduce $\bar{\boldsymbol{P}}^{k+1}$ and increase $\bar{c}^{P,k+1}$.
\end{enumerate}
The above scheme requires additional constraints on the change in price-quantity bids, $\Delta=(\delta c^{P,k}, \delta \boldsymbol{P}^k, \delta c^{G,k}, \delta \boldsymbol{q}^k)$. The actual strategy to choose $\Delta$ can vary. For example, a ``brute-force search'' would evaluate the profit for all possible $\Delta$ and pick the bid with the highest profit. Below, we employ a \emph{patternsearch} algorithm emulating this brute-force approach. 

\section{MARKET POWER EXAMPLES} \label{sec:examples}

\subsection{Electric Network Market Power}
Market power can be shown in the 3-node power network in Fig. \ref{fig:electic_network}. 
The branch impedance of each line is $2.81\cdot 10^{-4}$.  The load, branch thermal limit, max generator power, and generator cost varies across examples and is always labeled.  In Fig. \ref{fig:electic_network}, the load at bus 1 is higher than the power that can flow between bus 1 and  bus 3, where the lowest-cost generator sits.  Due to congestion on the branch, a higher cost generator at node 3 is used to deliver power. This congestion causes the LMPs to differ between nodes.  However, the existence of this congestion does not imply market power.

\begin{figure}[t]
    \centering
    \includegraphics[width=0.8\linewidth]{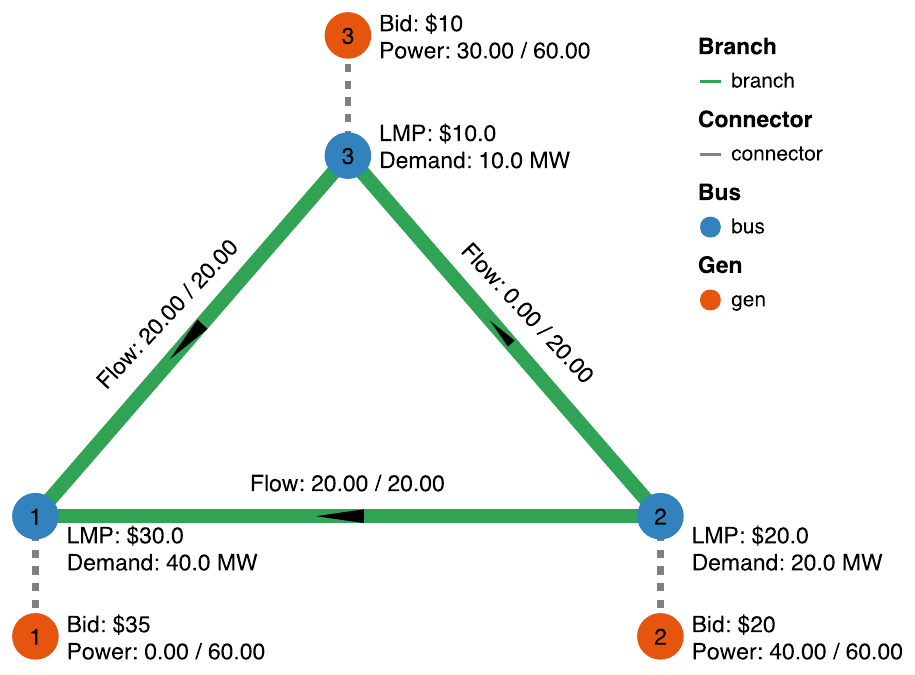}
    \caption{\small Power grid network with congestion.  The LMP at each node is set by the bid price of a generator at the node.  Each generator could bid higher to increase profit, but could not triple their bid or increase the bid by \$100 and still be cleared in the market.  Gen 1 and 2 do not have market power despite being in a congested region.}
    \label{fig:electic_network}
\end{figure}

Table \ref{tab:electric_scenarios} shows five additional scenarios on the three bus electric network, with small changes made to the network relative to the network shown in Fig. \ref{fig:electic_network}.  The first row shows that if generator 1 increased its bid price to \$100, there would be no impact on the LMPs because they can be served by a combination of generator 1 and 2, where the LMP at node 1 is \$30 because the more expensive generator at node 2 must run rather than the cheaper one at node 3.  If the load at node 1 is increased to 50 MW, then generator 1 has market power and can set the price, up to the conduct threshold.

\begin{table}[t]
\vspace{-2ex}
\caption{Market Power Scenarios} \label{tab:electric_scenarios}
\vspace{-2ex}
\begin{tabular}{@{}llll@{}}
\toprule
\multicolumn{1}{c}{\textbf{Network Changes}} & \textbf{LMP 1} & \textbf{LMP 2} & \textbf{LMP 3} \\ \midrule
Increase gen 1 price bid to \$100 & \$30 & \$20 & \$10 \\
Increase gen 1 bid and load 1 to 50 MW & \$100 & \$20 & \$10 \\
Increase all line limits to 30 & \$20 & \$20 & \$20 \\
Remove gen 2 & \$30 & \$50 & \$10 \\
Set line limit (1,2) to 10, load 2 to 5 MW & \$30 & \$-10 & \$10 \\ \bottomrule
\end{tabular}
\vspace{1ex}
\end{table}

The third scenario from Table \ref{tab:electric_scenarios} shows that if all line limits are increased, the congestion is removed.  Scenario four shows that if generator 2 is removed, the LMP at 2 can become higher than any generator bid price.  Scenario five shows that if flow on the line connecting nodes 1 and 2 is reduced and the load at 2 is reduced, the LMP at 2 can become negative, even when all bids are positive.   These scenarios demonstrate that congestion alone does not indicate market power, and that increased (or negative) LMP values can be caused by congestion without any generator changing their bid prices.  This is key to testing if a gas-fired generator is exerting market power across gas and electric networks.

\vspace{-1ex}
\subsection{Congestion in Coupled Gas and Electricity Networks }

\begin{table}[t]
\centering
\caption{Pipe Data} \label{tab:pipes}
\vspace{-2ex}
\begin{tabular}{@{}cccc@{}}
\toprule
Pipe & Diameter (meters) & Length (meters) & Friction Factor \\ \midrule
$2\rightarrow3$ & 0.6 & 5,000 & 0.025  \\
$4\rightarrow5$ & 0 & 8,000 & 0.025 \\ \bottomrule
\end{tabular}
\end{table}

\begin{table}[t]
\centering
\caption{Compressor Data} \label{tab:compressor}
\vspace{-2ex}
\begin{tabular}{@{}ccc@{}}
\toprule
Compressor & $\alpha^{max}$ & $c^C_{ij}$\\ \midrule
$1\rightarrow2$ & 1.4 & 10 \\
$3\rightarrow4$ & 1.4 & 10 \\ \bottomrule
\end{tabular}
\end{table}

Next, we consider a 5-node gas network coupled to the electrical network in Figure \ref{fig:electic_network}, as shown in Figure \ref{fig:normal_congestion}.  This network has two compressors and two pipes, whose properties are shown in Tables \ref{tab:pipes} and \ref{tab:compressor}.  The minimum and maximum pressure bounds at the nodes of the network are 3 MPa and 6 MPa. There is one gas supply point, and two locations where gas is delivered.   Delivery points 1 and 2 are generators 1 and 2 in the electric network, respectively.  We consider these generators to be owned by the same company. The price-quantity bids are shown in the figure.  We consider a linear heat-rate of $h^{-1}=7.5 ~\textrm{MMBtu/MWh}$ and a mass-to-energy conversion rate of $0.0556 ~\textrm{MMBtu/kg}$.

\begin{figure}[!t]
    \centering
        \includegraphics[width=.8\linewidth]{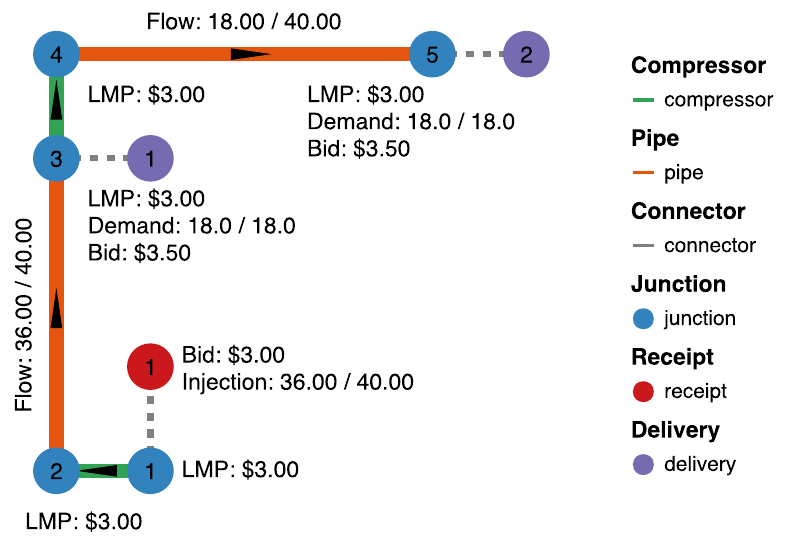}
        \includegraphics[width=.8\linewidth]{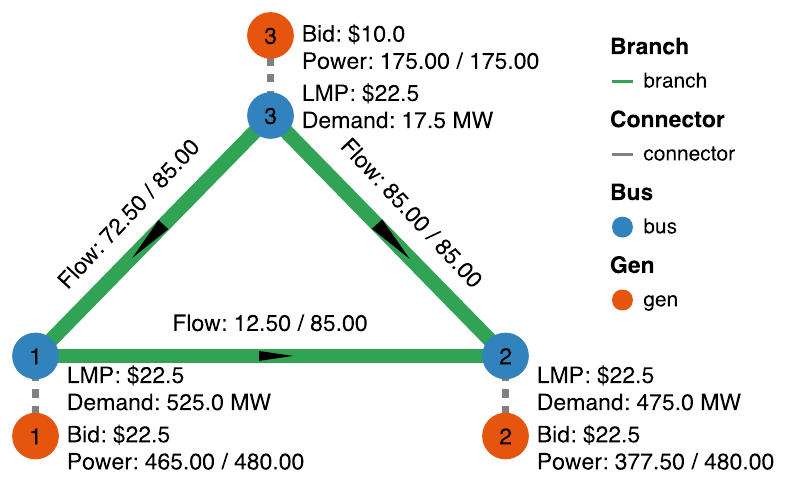}
    \caption{\small Gas (top) and power (bottom) networks.  Generator 2 is marginal, and produces at less than maximum capacity.}
    \label{fig:normal_congestion}
\end{figure}

Figure \ref{fig:normal_congestion} shows the networks with no congestion.  In the gas network, the generator at Delivery node 1 and 2 receive their full gas nomination. In the electrical network, generator 1 and 2 are both marginal, and produce less than their maximum power.  We use a pass-through electric bid price for simplicity, resulting in generators 1 and 2 making a small loss by selling slightly less electricity than they could generate from the gas they bought, shown in the first row of Table \ref{tab:small_gen_profits}.  In reality, the bid prices here would be higher to compensate for the costs of operating the generators.  These conditions represent normal operations. 

\begin{table}[!t]
    \centering
    \vspace{-1ex}
    \caption{Generator Profits}
    \label{tab:small_gen_profits}
    \vspace{-2ex}
    \begin{tabular}{@{}lrrr@{}}
    \textbf{Scenario} & \textbf{Generator 1} & \textbf{Generator 2} & \textbf{Total Profit} \\ \midrule
    Normal Operations & \$ -337.55 & \$ -2,306.3 & \$ -2,643.85 \\
    Collusion & \$ -5,775.06 & \$ 6,066.67 & \$ 291.61 \\
    Collusion and $c_{g_1}$ = \$80 & \$ 19,666.61 & \$ 46,200.00 & \$ 65,866.61
    \end{tabular}
    \vspace{1ex}
\end{table}

However, the generators can gain more profit. Figure \ref{fig:enhanced_congestion} shows the  networks with Delivery node 1 increasing its bid gas from 3.5 \$/MMBtu to 3.6 \$/MMBtu and quantity for gas from 19 MMBtu to 26 MMBtu, limiting gas flow to delivery node 2 (which gets only 14 of 19 MMBtu requested).  This reduces generation capacity at node 2, and increases the LMP at node 2 due to electrical network congestion. This results in much higher total profit for generators 1 and 2, despite buying extra gas for generator 1 that is not used to produce electricity. At the same time, the LMP at node 3 drops from \$26.2 to \$10, reducing profit for a competing generator.

\begin{figure}[t]
    \centering
        \includegraphics[width=.8\linewidth]{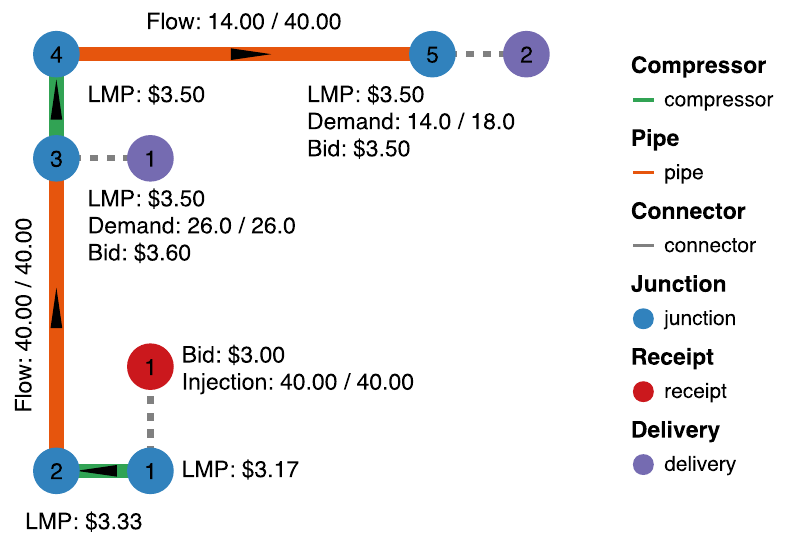}
        \includegraphics[width=.8\linewidth]{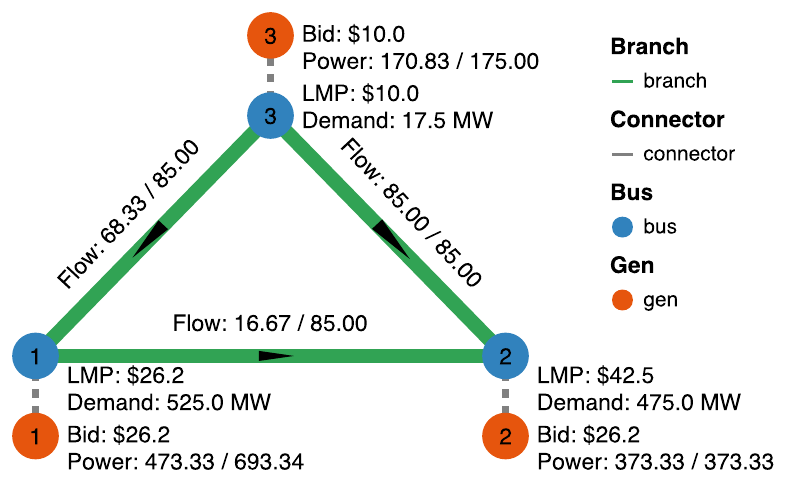}
    \label{fig:gaselec_2} 
    \vspace{1ex}
    \caption{\small Bidding a higher price and quantity in the gas network, generator 1 has more capacity in the electrical network, but the capacity at generator 2 is reduced. This increases the LMP at node 2 due to electrical network congestion, without generator 2 changing its price bid.}
    \label{fig:enhanced_congestion}
\end{figure}

Here generator 1 has market power and can increase its electric bid price to drive higher profits at generator 2. Due to congestion, the LMP at node 2 can be increased above market monitoring thresholds.  If the reference bid (average cleared bid in the last 90 days) is \$22.5, then generator 1 could bid up to \$90 without triggering a conduct test.

If generator 1 increases its electric bid price to \$80, the LMP at node 1 increases to \$80, the LMP at node 3 stays at \$10, and the LMP at node 2 becomes \$150.  The LMP at node 2 is higher than a bid that would otherwise trigger a conduct test, due to congestion in the electrical network due to gas network congestion.  Because no bid is high enough to trigger a conduct test, no mitigation is applied.  A generator can adjust gas bids to cause congestion in the power grid while avoiding a market power violation. 

Rather than constructing the bidding scenarios, similar insights can be obtained by solving the \emph{GPA} optimization formulation in Eq. \eqref{opt:GPA}. As shown in Tab. \ref{tab:3bus5node}, we see that single generator profit yields substantially lower LMPs (Scenarios (A) and (B) in Tab. \ref{tab:3bus5node}) in comparison to coordinated bidding in gas-market (Scenarios (C) and (D)) even though bid prices hit the chosen upper-bound of $\$100/\rm{MWh}$. 
This example is shown on two small network,  next we solve the GPA problem \eqref{opt:GPA} for larger networks.
\begin{table}[!t]
    \centering
    \vspace{1ex}
    \caption{Optimal bids $(\bar{\boldsymbol{P}}, \bar{c}^P; \bar{\boldsymbol{q}}, \bar{c}^G)$  for the GPA problem \eqref{opt:GPA} for various generator sets. Scenario (D) is same as (C) with different initial guess.
    }
    \label{tab:3bus5node}
    \vspace{1ex}
    \begin{tabular}{@{}l|rrr|r@{}}
         \textbf{Scenario} &\textbf{Gen. 1} & \textbf{Gen. }2 & \textbf{Gen. 3 }&\textbf{Profit}\\
         \midrule
         \textbf{(A)} Bids &$\bmat{355&\$100\\13.3&\$5}$ & & &\\          %
         \qquad \textit{LMPs} &$\$100$&$\$20$&$\$10$& $\$27.5K$\\\midrule
         \textbf{(B)} Bids & &$\bmat{305&\$100\\11.4&\$4}$ & &\\          %
         \qquad \textit{LMPs}&$\$30$&$\$100$&$\$10$&$\$23.6K$\\\midrule
         \textbf{(C)} Bids &$\bmat{694&\$100\\26&\$3.6}$ & $\bmat{374&\$26.2\\23&\$3}$& &\\          %
         \qquad \textit{LMPs}&$\$100$&$\$190$&$\$10$&$\$94.2K$\\\midrule
         \textbf{(D)} Bids &$\bmat{381&\$100\\14.3&\$3.5}$ & $\bmat{507&\$100\\19&\$3.5}$& &\\          %
         \qquad \textit{LMPs}&$\$190$&$\$100$&$\$10$&$\$103.2K$
    \end{tabular}
    \vspace{3ex}
\end{table}

\begin{figure}[!b]
    \centering
    \includegraphics[width=0.75\linewidth]{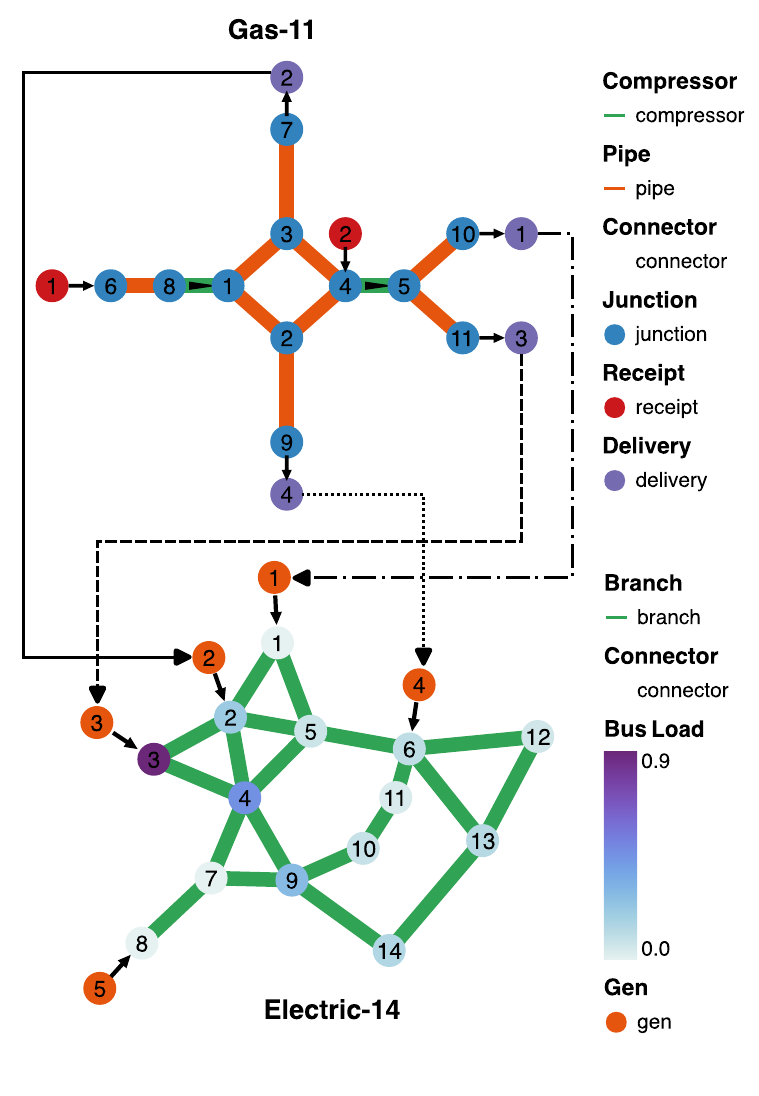}
    \vspace{-2ex}
    \caption{\small Coupling of IEEE-14 and GasLib-11 cases. Arrows indicate generators connections to gas and power nodes.}
    \label{fig:network14-11}
\end{figure}

\subsection{IEEE-14 and GasLib-11 Coupled System}

Consider the Gaslib-11 network \cite{schmidt2017gaslib} coupled with \emph{IEEE14} of MATPOWER \cite{zimmerman1997matpower} where the 4 gas-fired generators at buses 1, 2, 3, and 6 act as buyers in the gas network at demand nodes 10, 7, 11, and 9, respectively. See Fig. \ref{fig:network14-11}.  We modify the IEEE-14 case by changing maximum generation capacity and  cost (MWh, \$/MWh) at nodes 1, 2, 3, 6, and 8 to $(200,40)$, $(140,20)$, $(100,25)$, $(100,30)$, and $(100,15)$, respectively.  The GasLib-11 case is modified as follows: (a) horizontal pipes, 0.01 friction factors, diameter 0.3m; and (b) $s_m^{\max}$ at nodes 6 and 4 are 10 \textrm{kg/s} and 8 \textrm{kg/s}, respectively.

Solving the GPA problem \eqref{opt:GPA} yields bidding strategies that are difficult to interpret as anti-competitive due to the complexity of the network interactions, and this becomes even harder for larger networks (see Tab. \ref{tab:14bus11node}).  While identifying anti-competitive bidding is non-trivial, nonetheless, we observe that LMPs are significantly influenced by these strategies as seen in Tab. \ref{tab:14bus11node}, where the baseline case has no coordinated bidding whereas Scenarios (A) and (B) depict coordination between generators $(2,3)$ and $(2,4)$, respectively. Whether a subset of generators is colluding is significant and appropriate monitoring mechanisms are necessary to limit monopolistic market power.

\begin{table}[t]
    \centering
    \caption{Optimal bids for IEEE-14 \& GasLib-11 systems}
    \label{tab:14bus11node}
    \begin{tabular}{@{}l|rrrr@{}}
       \textbf{Scenario} &\textbf{Gen. 1} & \textbf{Gen. }2 & \textbf{Gen. 3 }&\textbf{Gen. 4}\\
         \midrule
         \textbf{Baseline}&&&&\\
         \quad \textit{LMPs}&$\$25$&$\$25$&$\$25$&$\$25$\\
         \midrule
         \textbf{(A)} Bids &&$\bmat{59\!\!&\$20\\2.2\!\!&\$2.75}$ &$\bmat{100\!\!&\$25\\1.2\!\!&\$3}$ & \\     %
         \quad \textit{LMPs} &$\$30$&$\$30$&$\$30$& $\$30$\\\midrule
         \textbf{(B)} Bids &&$\bmat{22\!\!&\$20\\0.8\!\!&\$2.75}$ &&$\bmat{37\!\!&\$30\\1.4\!\!&\$3.25}$ \\ %
         \quad \textit{LMPs} &$\$40$&$\$40$&$\$40$& $\$40$\\\midrule
    \end{tabular}
    \vspace{2ex}
\end{table}

\section{CONCLUSIONS}  \label{sec:conc}

We have shown how a gas-fired generator can manipulate electricity markets by adjusting its bid into natural gas markets while avoiding detection by power system market monitoring.  We define market monitoring for anti-competitive behavior in DC-OPF cleared power grid markets, present a linear-programming method for detecting market power by a set of generators, and develop an optimization model for a generator to maximize their profit through bids into gas and electric markets.  Using small, interpretable networks, we show how market power can be used to increase generator profits across gas and power systems. Generators can collude to modify bids in the gas market to profit from congestion and high LMPs in the electrical market that would otherwise be mitigated due to monitoring.  A heuristic is tested for a larger network to find profit maximizing bids.  Future work will extend the market power certificate in Proposition \ref{prop:marketPowerTest} to the gas network case, to enable market monitoring across gas and power systems.

\bibliographystyle{unsrt}
\bibliography{references}

\begin{thebibliography}{10}

\bibitem{he2018coordination}
Chuan He, Xiaping Zhang, Tianqi Liu, Lei Wu, and Mohammad Shahidehpour.
\newblock Coordination of interdependent electricity grid and natural gas network—a review.
\newblock {\em Current Sustainable/Renewable Energy Reports}, 5(1):23--36, 2018.

\bibitem{guerra2021coordinated}
Omar~J. Guerra et~al.
\newblock Coordinated operation of electricity and natural gas systems from day-ahead to real-time markets.
\newblock {\em Journal of cleaner production}, 281:124759, 2021.

\bibitem{craig2020valuing}
Michael Craig et~al.
\newblock Valuing intra-day coordination of electric power and natural gas system operations.
\newblock {\em Energy Policy}, 141:111470, 2020.

\bibitem{zlotnik2016coordinated}
Anatoly Zlotnik, Line Roald, et~al.
\newblock Coordinated scheduling for interdependent electric power and natural gas infrastructures.
\newblock {\em IEEE Trans. on Power Sys.}, 32(1):600--610, 2016.

\bibitem{duenas2014gas}
Pablo Due{\~n}as et~al.
\newblock Gas--electricity coordination in competitive markets under renewable energy uncertainty.
\newblock {\em IEEE Trans. Power Systems}, 30(1):123--131, 2014.

\bibitem{bunn1997modelling}
Derek~W. Bunn, Isaac Dyner, et~al.
\newblock Modelling latent market power across gas and electricity markets.
\newblock {\em System Dynamics Review: The Journal of the System Dynamics Society}, 13(4):271--288, 1997.

\bibitem{nandakumar2016computational}
Neha Nandakumar.
\newblock {\em Computational models of natural gas markets for gas-fired generators}.
\newblock PhD thesis, Massachusetts Inst. Tech., 2016.

\bibitem{leisen2024modeling}
Robin Leisen et~al.
\newblock Modeling combined-cycle power plants in a detailed electricity market model.
\newblock {\em Energy}, 298:131246, 2024.

\bibitem{levin2015electricity}
Todd Levin and Audun Botterud.
\newblock Electricity market design for generator revenue sufficiency with increased variable generation.
\newblock {\em Energy Policy}, 87:392--406, 2015.

\bibitem{guo2018market}
Hongye Guo et~al.
\newblock Market equilibrium analysis with high penetration of renewables and gas-fired generation: An empirical case of the {Beijing-Tianjin-Tangshan} power system.
\newblock {\em Applied Energy}, 227:384--392, 2018.

\bibitem{jiang2021bi}
Tao Jiang et~al.
\newblock Bi-level strategic bidding model of gas-fired units in interdependent electricity and natural gas markets.
\newblock {\em IEEE Trans. on Sustainable Energy}, 13(1):328--340, 2021.

\bibitem{mohlin2021us}
Kristina Mohlin.
\newblock The {US} gas pipeline transportation market: An introductory guide with research questions for the energy transition.
\newblock {\em Environmental Defense Fund EDP}, pages 21--01, 2021.

\bibitem{hulshof2016market}
Daan Hulshof et~al.
\newblock Market fundamentals, competition and natural-gas prices.
\newblock {\em Energy policy}, 94:480--491, 2016.

\bibitem{jenkins2015dynamic}
Sandra Jenkins et~al.
\newblock A dynamic model of the combined electricity and natural gas markets.
\newblock In {\em Power \& Energy Society Innovative Smart Grid Technologies Conference (ISGT)}, pages 1--5. IEEE, 2015.

\bibitem{spiecker2011analyzing}
Stephan Spiecker.
\newblock Analyzing market power in a multistage and multiarea electricity and natural gas system.
\newblock In {\em 8th International Conf. on the European Energy Market}, pages 313--320. IEEE, 2011.

\bibitem{spiecker2013modeling}
Stephan Spiecker.
\newblock Modeling market power by natural gas producers and its impact on the power system.
\newblock {\em IEEE Trans. Power Sys.}, 28(4):3737--3746, 2013.

\bibitem{morey2025winners}
Danielle~F Morey, Michelle Fischer, and Ramteen Sioshansi.
\newblock Winners and losers from vertical integration between natural-gas and electricity markets.
\newblock {\em The Energy Journal}, page 01956574251324175, 2025.

\bibitem{creti2016natural}
Anna Cret{\`\i} and Federico Pontoni.
\newblock Natural gas balancing, storage, and flexibility in europe: Assessing the recent literature.
\newblock {\em Current Sustainable/Renewable Energy Reports}, 3(1):18--22, 2016.

\bibitem{TGP_EntireTariff}
{Tennessee Gas Pipeline Company, L.L.C.}
\newblock Currently effective rates.
\newblock https://pipeline2.kindermorgan.com/Tariff/ SubIndex.aspx?code=TGP\&category=CER.

\bibitem{keyaerts2011gas}
Nico Keyaerts, Michelle Hallack, et~al.
\newblock Gas market distorting effects of imbalanced gas balancing rules: Inefficient regulation of pipeline flexibility.
\newblock {\em Energy Policy}, 39(2):865--876, 2011.

\bibitem{szolnoki2017monitoring}
P{\'a}lma Szolnoki.
\newblock {\em Monitoring Natural Gas Balancing Markets}.
\newblock PhD thesis, Corvinus University of Budapest, 2017.

\bibitem{graf2021market}
Christoph Graf, Emilio La~Pera, Federico Quaglia, and Frank~A Wolak.
\newblock Market power mitigation mechanisms for wholesale electricity markets: Status quo and challenges.
\newblock {\em Work. Pap. Stanf. Univ}, 2021.

\bibitem{CAISO_mm}
CAISO.
\newblock Caiso market monitoring.
\newblock https://www.caiso.com/market-operations/market-monitoring.

\bibitem{pjm2025som}
{Monitoring Analytics, LLC}.
\newblock {State of the Market Report for {PJM}, January through September 2025}.

\bibitem{NYISO2023SOMreport}
David~B. Patton et~al.
\newblock {State of the Market Report for the New York ISO Markets}, 2023.

\bibitem{MISO_MPB_Energy_Market}
MISO.
\newblock Energy and operating reserve markets.
\newblock {\em Business Practices Manual Energy and Operating Reserve Markets}, 2024.

\bibitem{ercot2025mitigation}
ERCOT.
\newblock {1255NPRR: Introduction of Mitigation of ESRs 100224}, 2024.
\newblock https://www.ercot.com/mktrules/issues/NPRR1255.

\bibitem{goldman2004review}
Charles Goldman, Bernie~C. Lesieutre, and Emily Bartholomew.
\newblock A review of market monitoring activities at {US} independent system operators.
\newblock {\em Lawrence Berkeley National Lab.}, (LBNL-53975), 2004.

\bibitem{garcia2007international}
Jos{\'e}~A Garc{\'\i}a and James~D Reitzes.
\newblock International perspectives on electricity market monitoring and market power mitigation.
\newblock {\em Review of Network Economics}, 6(3), 2007.

\bibitem{rahimi2003effective}
A.~F. Rahimi et~al.
\newblock Effective market monitoring in deregulated electricity markets.
\newblock {\em IEEE Trans. Power Sys.}, 18(2):486--493, 2003.

\bibitem{adib2008market}
Parviz Adib and David Hurlbut.
\newblock Market power and market monitoring.
\newblock In {\em Competitive Electricity Markets}, pages 267--296. Elsevier, 2008.

\bibitem{peterson2001best}
Paul Peterson et~al.
\newblock Best practices in market monitoring.
\newblock {\em A report prepared for Maryland Office of People’s Counsel}, 2001.

\bibitem{sun2026negative}
Qingkai Sun et~al.
\newblock Negative electricity prices in electricity markets: causes, impacts, and response strategies.
\newblock {\em Global Energy Interconnection}, 2026.

\bibitem{barbour2014can}
Edward Barbour et~al.
\newblock Can negative electricity prices encourage inefficient electrical energy storage devices?
\newblock {\em Internat. J. Environmental Studies}, 71(6):862--876, 2014.

\bibitem{stiewe2025cross}
Clemens Stiewe et~al.
\newblock Cross-border cannibalization: Spillover effects of wind and solar energy on interconnected european electricity markets.
\newblock {\em Energy Economics}, 143:108251, 2025.

\bibitem{twomey2005review}
Paul Twomey.
\newblock A review of the monitoring of market power: The possible roles of tsos in monitoring for market power issues in congested transmission systems.
\newblock 2005.

\bibitem{caramanis2007optimal}
Michael~C Caramanis, Roger~E Bohn, and Fred~C Schweppe.
\newblock Optimal spot pricing: Practice and theory.
\newblock {\em IEEE Transactions on Power Apparatus and Systems}, (9):3234--3245, 2007.

\bibitem{rios2015optimization}
Roger~Z. R{\'\i}os-Mercado and Conrado Borraz-S{\'a}nchez.
\newblock Optimization problems in natural gas transportation systems: A state-of-the-art review.
\newblock {\em Applied Energy}, 147:536--555, 2015.

\bibitem{rudkevich2017locational}
Aleksandr~M. Rudkevich et~al.
\newblock {Locational Marginal Pricing of Natural Gas subject to Engineering Constraints}.
\newblock In {\em 50th Hawaii International Conference on System Sciences}, pages 3092--3101, 2017.

\bibitem{schmidt2017gaslib}
Martin Schmidt et~al.
\newblock Gaslib—a library of gas network instances.
\newblock {\em Data}, 2(4):40, 2017.

\bibitem{zimmerman1997matpower}
Ray~D. Zimmerman et~al.
\newblock Matpower.
\newblock {\em PSERC.[Online]. Software Available at: http://www.pserc.cornell.edu/matpower}, 1997.

\end{thebibliography}

\end{document}